\documentclass[letterpaper,12pt]{article}
\usepackage{amssymb, amsmath}
\usepackage{amsthm}

\newtheorem{theorem}{Theorem}
\newtheorem{corollary}[theorem]{Corollary}
\newtheorem{fact}[theorem]{Fact}
\newtheorem{lemma}[theorem]{Lemma}
\newcommand{\comment}[1]{}

\begin{document}

\title{On reversible cascades in scale-free and
Erd\H{o}s-R\'enyi random
graphs}

\author{
Ching-Lueh Chang
\footnote{Department of Computer Science and Engineering, Yuan Ze University,
Taoyuan, Taiwan. Email:
clchang@saturn.yzu.edu.tw}
}

\maketitle

\begin{abstract}
Consider the
following
cascading
process
on a simple undirected graph $G(V,E)$
with diameter $\Delta$.
In round zero,
a set $S\subseteq V$ of vertices, called the seeds,
are active.
In
round $i+1,$
$i\in\mathbb{N},$
a non-isolated vertex
is activated
if at least
a $\rho\in(\,0,1\,]$
fraction of its neighbors
are active in round $i$;
it is deactivated otherwise.
For $k\in\mathbb{N},$ let
$\text{min-seed}^{(k)}(G,\rho)$
be the minimum number of seeds
needed to activate all vertices
in
or before
round $k$.
This paper
derives
upper bounds on $\text{min-seed}^{(k)}(G,\rho)$.
In particular,
if
$G$ is connected and
there exist constants $C>0$ and $\gamma>2$
such that the fraction of degree-$k$ vertices in $G$
is at most $C/k^\gamma$ for all $k\in\mathbb{Z}^+,$
then
$\text{min-seed}^{(\Delta)}(G,\rho)
=O(\lceil\rho^{\gamma-1}\,|\,V\,|\rceil)$.
Furthermore,
for
$n\in\mathbb{Z}^+,$
$p=\Omega((\ln{(e/\rho)})/(\rho n))$
and with probability $1-\exp{(-n^{\Omega(1)})}$
over the Erd\H{o}s-R\'enyi random graphs $G(n,p),$
$\text{min-seed}^{(1)}(G(n,p),\rho)
=O(\rho n)$.
\comment{
The bounds
carry over to all
synchronizations of
the cascading process
that do not
simply
refuse
all
valid
changes of states
forever.
}
\end{abstract}

\section{Introduction}\label{introductionsection}

Let $G(V,E)$ be a simple directed graph, $\rho\in(\,0,1\,]$
and $S\subseteq V,$
where each vertex
of $G$
can be in one of two states,
active or inactive.
The
synchronous reversible cascade
proceeds
in rounds.
In round
zero,
only
the vertices in $S,$ called the seeds, are active.
In round $i+1,$
a vertex
with a positive indegree is activated (resp., deactivated)
if at least (resp., less than) a $\rho$ fraction of its in-neighbors are active
in round $i,$ where $i\in \mathbb{N}$.
For $k\in\mathbb{N},$
define $\text{min-seed}^{(k)}(G,\rho)$ to be the
minimum number of seeds needed
so that
all vertices
will be active
in
or before
round $k$.

The synchronous reversible cascade above
is the same as
the
local interaction game with full rationality
except that
the latter is defined on infinite graphs with finite degrees~\cite{Mor00,
Kle07}.
For local interaction games,
Morris~\cite{Mor00}
studies
conditions allowing
finitely many
seeds
to
activate
each vertex
sooner or later.
Blume~\cite{Blu93}, Ellison~\cite{Ell93}, Young~\cite{You98, You06}
and Montanari and Saberi~\cite{MS09}
study
a
variant where
there is
another variable
specifying
the vertices' degree of rationality
and
the states are updated according to a Poisson clock.
In their
fully rational
scenario,
updating
the state of
a vertex $v$
means
activating or deactivating $v,$
respectively,
if at least or less than a certain fraction of $v$'s neighbors
are active.
They analyze the expected waiting time
until
all or most vertices
enter the same state.

Consider the special cases of
the
synchronous reversible cascade
with $\rho=1/2$
or $\rho=1/2+1/(2\,|\,V\,|)$.
So
a vertex
is activated
in a round
if
the
simple
or strict
majority of its in-neighbors
are active in the previous round; it is deactivated otherwise.
Both
special cases,
among other similar
dynamics,
are suitable for modeling transient faults
in majority-based
fault-tolerant
systems~\cite{Pel02, FKRRS03, FLLPS04}.
Call a set of seeds
a $k$-round monopoly if it activates all vertices in
or before
round $k,$
where $k\in\mathbb{N}$.
Any
finite-round monopoly
is called a dynamic monopoly.
If a
set of seeds
does not lead to the deactivation of
active vertices
in any round, then it is said to be monotone.
Peleg~\cite{Pel98} shows an $\Omega(\sqrt{|\,V\,|}\,)$ lower bound
on the minimum size of $2$-round monopolies in any simple undirected
graph $G(V,E)$.
Bounds
are known on the minimum size of
monotone dynamic monopolies
in planar graphs~\cite{FKRRS03}, toroidal meshes, torus cordales,
torus serpentini~\cite{FLLPS04}
and simple undirected graphs~\cite{Pel98}.
Berger~\cite{Ber01}
shows the existence of constant-size dynamic monopolies
in an infinite family of simple undirected graphs.
Optimal bounds on the
minimum size of $1$-round monopolies
are also derived
for planar graphs, hypercubes, graphs with a given
girth, graphs with diameter $2$~\cite{LPRS93}
and simple undirected graphs~\cite{LPRS93, Pel96b}.
Bermond et al.~\cite{BBPP03}
derive bounds on the minimum size of $1$-round monopolies
under the variant where a vertex $v$ is activated
or deactivated, respectively,
if the majority or minority of the vertices
within distance $r$ from $v$ are active,
$r\ge 2$.
Peleg~\cite{Pel02} surveys
the above and
many
related results.

As an
important variant,
the irreversible
cascade
on a graph
forbids the deactivation of vertices.
In its most general form,
a vertex $v$ is activated
in a round
if
at least $\phi(v)$
of
its in-neighbors are active
in the previous round, where $\phi(v)\in\mathbb{N}$.
It is
sometimes defined as an asynchronous process
because
the order of activating the vertices
does not affect the set of vertices that will be
active.
With various threshold functions $\phi(\cdot),$
irreversible
cascades
describe
the propagation of permanent faults in majority-based
systems~\cite{LPS99, FGS01, FKRRS03, FLLPS04},
spread of diseases~\cite{DR09},
complex propagation~\cite{CEM07, GC07, SS06},
socio-economic
cascades~\cite{Gra78, Wat02}
and
cascading failures in infrastructure or organizational networks~\cite{Wat02}.
Assuming
irreversible cascades,
bounds
are derived
on the minimum
number of seeds activating all vertices
eventually
for
complete trees, rings, butterflies, wrapped butterflies,
cube-connected cycles, shuffle-exchange graphs, DeBruijn graphs,
hypercubes~\cite{LPS99, FKRRS03},
toroidal meshes~\cite{Luc98, FLLPS04, PZ05, KLV09, DR09},
torus cordales, torus serpentini~\cite{Luc98, FLLPS04, DR09},
chordal rings~\cite{FGS01},
multidimensional cubes~\cite{BP98},
complete multipartite
graphs, regular graphs~\cite{DR09},
Erd\H{o}s-R\'enyi random graphs~\cite{CL09, CLTOCS, Cha10},
undirected connected graphs~\cite{CL10CIAC}
and
directed graphs with positive indegrees~\cite{CL10CIAC, ABW10}.
More
bounds are derived on
the minimum number of seeds
activating all vertices
in one round
for near-regular graphs,
bounded-degree graphs~\cite{Pel96b}
and, under several variants,
simple undirected graphs~\cite{BBPP03}.

\comment{
Peleg~\cite{Pel96b}
derives optimal bounds
on the
minimum number of seeds activating all vertices
in $1$ round
for near-regular and bounded-degree graphs.
Bermond et al.~\cite{BBPP03}
study
the same
quantity
in simple undirected graphs
under
variants
where the state
of a vertex is determined
by majority votes among those within distance $r$ from it,
$r\in\mathbb{Z}^+$.
}

A family of undirected graphs, $\{G_n(V_n,E_n)\mid |\,V_n\,|=n\}_{n=1}^\infty,$
is scale-free if there
is a constant
$2<\gamma<3$
such that
for sufficiently large $k\in\mathbb{Z}^+$
and as $n\to\infty,$
the fraction of
vertices with degree
$k$
in $G_n$
is proportional to $1/k^\gamma$.
Directed scale-free graphs are defined similarly
by using
the in- or outdegrees of vertices
instead.
Many
socio-economic, physical, biological and semantic networks
are scale-free.
Generative models
of scale-free networks
include
preferential attachment
models~\cite{ACL02, BA99, AB02, BRST01, MX06},
copying models~\cite{BJ04a, BJ04b, KKRRT99, KRRSTU00},
growth-deletion models~\cite{CFV03, CL03, FFV07},
random-surfer models~\cite{BCR06, CM08},
traffic-driven model~\cite{BBV04},
heuristically optimized trade-off model~\cite{FKP02},
hybrid models~\cite{PFLGG02, CL04b, PRU06},
semantic growth model~\cite{ST05}
and
random-graph models with given expected degree sequences~\cite{CL02, CL04a}.
Many of these suitably describe the Webgraph~\cite{Bon04, DLLM04}.


Let $G(V,E)$ be an undirected connected graph
with diameter $\Delta,$
$\gamma>2$ be a constant and $\rho\in (\,0,1\,]$
such that the fraction of vertices with degree $k$
in $G$ is $O(1/k^\gamma)$.
This
paper proves
$\text{min-seed}^{(\Delta)}(G,\rho)=O(\lceil\rho^{\gamma-1}\,|\,V\,|\rceil)$.
As
scale-free graphs typically have
small distances between vertices~\cite{NSW01, CH03, BR04, CL04a, CL04b},
activating all vertices within $\Delta$ rounds
may be
fast.
Furthermore, the $O(\lceil\rho^{\gamma-1}\,|\,V\,|\rceil)$ bound
continues to hold
even
if
the
synchronous
reversible cascade
is modified to proceed asynchronously instead.

For $n\in\mathbb{Z}^+$ and $p\in[\,0,1\,],$ the
Erd\H{o}s-R\'enyi random graph $G(n,p)$
is a simple undirected graph with
vertices $1,2,\ldots,n$
where each of the $\binom{n}{2}$ possible edges
appears independently with probability $p$~\cite{Bol01}.
Assuming irreversible cascades,
Chang and Lyuu~\cite{CLTOCS, Cha10}
consider the
case
where
a
non-isolated
vertex of $G(n,p)$
is activated
when
at least a $\rho\in(\,0,1\,]$ fraction of its neighbors are active.
They
prove
that
for
a sufficiently large constant $\beta>0,$
$n\in\mathbb{Z}^+,$
$\delta\in\{1/n,2/n,\ldots,n/n\},$
$p\ge \beta(\ln(e/{\min\{\delta,\rho\}}))/(\rho n)$
and with probability $1-n^{-\Omega(1)}$
over $G(n,p),$
the minimum number of seeds
needed to
eventually
activate
at least $\delta n$
vertices
is $\Theta(\min\{\delta,\rho\}\,n)$.
The hidden constant in the $\Theta(\cdot)$ notation
is independent of $p$.
For
a sufficiently large constant $\lambda>0,$
$n\in\mathbb{Z}^+,$
$p\ge\lambda(\ln(e/\rho))/n$
and with probability $1-n^{-\Omega(1)}$ over $G(n,p),$
they also prove
the existence of
$O(\lceil\rho n\rceil)$ seeds
that
activate
all vertices eventually.
\comment{
With
more
careful calculations in their proofs,
$\lambda$ can in fact be
any constant
larger than $1$~\cite{CL10submitted}.
}

This paper proves that
for
$n\in\mathbb{Z}^+,$
$p=\Omega((\ln(e/\rho))/(\rho n))$ and
with probability $1-\exp{(-n^{\Omega(1)})}$
over $G(n,p),$
$\text{min-seed}^{(1)}(G(n,p),\rho)=O(\rho n)$.
Together with
Chang and Lyuu's $\Theta(\min\{\delta,\rho\}\,n)$ bound
for irreversible cascades with an
unbounded duration,
our result shows
that
neither
the
reversibility
of
synchronous
cascades
nor the number of rounds played
can change the asymptotically minimum number of seeds needed to
activate all vertices for
$p\ge \beta(\ln(e/{\min\{\delta,\rho\}}))/(\rho n)$
with a sufficiently large constant $\beta$ --- it
is always $\Theta(\rho n)$.
Furthermore, our $O(\rho n)$ bound
continues to hold
even if
the
synchronous
reversible cascade
is modified to proceed asynchronously instead.

\comment{
Furthermore, an
irreversible
$k$-round (resp., dynamic)
monopoly
is
a set of seeds activating all vertices
in or before round $k$ (resp., a finite number of rounds)
with
majority thresholds.
}

\comment{
a
complete tree, ring, butterfly, wrapped butterfly,
cube-connected cycle, shuffle-exchange, DeBruijn,
hypercube~\cite{LPS99, FKRRS03},
toroidal mesh~\cite{Luc98, FLLPS04, PZ05, KLV09, DR09},
torus cordalis, torus serpentinus~\cite{Luc98, FLLPS04, DR09},
chordal ring~\cite{FGS01},
complete multipartite
graph, regular graph~\cite{DR09},
Erd\H{o}s-R\'enyi random graph~\cite{CL09, CLTOCS, Cha10}
and
directed graphs with positive indegrees.
}

\comment{
Under the reversible model,
optimal
bounds are
derived
on
the
minimum size of $1$-round
monopolies
in
planar graphs, hypercubes, graphs with a given girth,
graphs with diameter $2$~\cite{LPRS93}
and general graphs~\cite{LPRS93, Pel96b}.
Assuming irreversible cascades,
instead,
Peleg~\cite{Pel96b} derives
optimal bounds for
near-regular and bounded-degree graphs.
Bermond et al.~\cite{BBPP03}
study the variant where
the state (i.e., active or inactive)
of
a vertex
is determined by majority votes among those within
distance $r$ from it.
For various values of $r,$ they
give optimal or nearly optimal
bounds
on the minimum size of $1$-round
monopolies.
Peleg~\cite{Pel02} surveys
the above and
many
related results.
}


Other related topics
include
periodic behavior of synchronous reversible
cascades~\cite{GFP85, GO80, PS83, PT86, GH00, Mor94a, Mor94b, Mor95},
inapproximability of minimum-seed problems~\cite{Che08, CL10CIAC, ABW10},
majority consensus computers~\cite{MP01},
stable sets of active vertices~\cite{Agu87, AFK88, Agu91, Gra91, Kra01},
network decontamination~\cite{LPS07, LP07},
maximization of social influence~\cite{DR01, RD02, KKT03, KKT05, MR07}
and percolation theory~\cite{AS94},
among others.

This paper is organized as follows.
Sec.~\ref{definitionssection}
presents
the notations
and the preliminaries.
Sec.~\ref{generalboundssection}
derives
general bounds
on the minimum number of seeds needed to activate
all vertices in a number of rounds.
Secs.~\ref{connectedscalefreegraphssection}--\ref{ErdosRenyirandomgraphssection}
investigate the cases of connected scale-free
and Erd\H{o}s-R\'enyi random graphs, respectively.
Sec.~\ref{conclusionssection}
concludes the paper.

\section{Definitions}\label{definitionssection}

A
directed graph
$G(V,E)$
consists of
a set $V$ of vertices and a set
$E\subseteq V\times V$
of edges.
An edge $(u,v)\in E$ goes from $u$ to $v$.
An
undirected graph is a directed one
with each edge accompanied by the edge in the opposite direction.
Unless otherwise specified,
all graphs in this paper are simple, i.e.,
self-loops are not allowed~\cite{Wes01}.
For $v\in V,$ define
\begin{eqnarray*}
N^\text{in}(v)&\equiv&\left\{u\in V\mid (u,v)\in E\right\},\\
N^\text{in}[v]&\equiv& N^\text{in}(v)\cup\{v\}
\end{eqnarray*}
to be
the open and closed in-neighborhoods of $v,$
respectively.
The indegree of $v$ is $d^\text{in}(v)\equiv |\,N^\text{in}(v)\,|$.
For $A\subseteq V,$
define $N^\text{in}(A)\equiv \cup_{a\in A}\, N^\text{in}(a)$
as
the set of vertices incident on an edge coming into $A$.
In case $G$ is undirected,
write
$N^\text{in}(\cdot),$ $N^\text{in}[\cdot]$
and $d^\text{in}(\cdot)$ simply as
$N(\cdot),$ $N[\cdot]$
and $d(\cdot),$ respectively.
For an undirected graph $G(V,E),$
$U\subseteq V$ and $i\ge 0,$
define $N^i[U]$ to be the set of vertices
with distance less than or equal to $i$
from at least one vertex in $U$.
That is, $N^0[U]=U$
and $N^{i+1}[U]= N^i[U]\cup N(N^i[U])$ for $i\ge 0$.
Furthermore,
each vertex
can be in one of two states, active or inactive.

The
synchronous reversible cascade
on $G(V,E)$
with
seed set $S\subseteq V$ and threshold $\rho\in (\,0,1\,]$
proceeds in rounds.
In round zero,
only the vertices in $S,$ called the seeds, are active.
For each $i\in\mathbb{N},$
a vertex with a positive indegree
is activated or deactivated in round $i+1,$
respectively,
if at least or less than a $\rho$ fraction
of its in-neighbors are active in round $i$.
Vertices with indegree zero, instead,
never change their states.
For $k\in\mathbb{N},$
define $\text{\sf Active}^{(k)}(S,G,\rho)$
to be
the set of active vertices in round $k$.
Then
define
$$\text{min-seed}^{(k)}\left(G,\rho\right)\equiv
\min_{W\subseteq V, \text{\sf Active}^{(k)}(W,G,\rho)=V}\,
|\,W\,|,$$
which is
the
minimum number of seeds needed to activate all vertices
in round $k$ (it is possible that a set of seeds activates
all vertices in round $k$ by doing so before round $k$).
Clearly, once all vertices are active, they will remain active
forever.
Therefore, $\text{min-seed}^{(k)}(\cdot,\cdot)$ monotonically decreases
as $k$ increases.

Next,
we
describe
an
asynchronous
reversible cascade.
At any instant,
one or more
vertices
may
update their states.
When
a vertex with a positive indegree
updates its state,
it is activated or deactivated
if at least or less than a $\rho$ fraction of
its in-neighbors are active, respectively.
Instead,
vertices
with indegree zero
never
update their states.
The only requirement on
the
synchronization
of updates
is that,
if the state of a vertex can change at time $t,$
then
at least one vertex must change its state after time $t,$ where
$t\ge 0$.
Such a synchronization is said to be progressive.
In other words,
whenever changes of states are possible,
an asynchronous
reversible cascade
cannot simply
refuse
them
all
forever.
Define $\text{min-seed}^\text{async}(G,\rho)$
to be the minimum number of seeds
that activate all vertices within a finite amount of time
in
every
asynchronous
reversible
cascade
with threshold $\rho$.
So
$\text{min-seed}^\text{async}(G,\rho)\le s$
precisely when
there exist $s$ seeds
activating
all vertices
within a finite amount of time
no matter how
the updates of states
are synchronized in a progressive way.
Furthermore,
define $\text{min-seed}^\text{async}_\text{monotone}(G,\rho)$
to be the minimum number of seeds
meeting the following criteria
in
every
asynchronous reversible
cascade
with threshold $\rho$:
\begin{itemize}
\item All vertices are active after a finite amount of time;
\item No active vertices are ever deactivated.
\end{itemize}
So $\text{min-seed}^\text{async}_\text{monotone}(G,\rho)\le s$
precisely when
there exist $s$ seeds
that, regardless of the (progressive) synchronizations
of the reversible cascades,
activate all vertices without ever deactivating any active
one.

The following
fact
is
folklore.
See, e.g.,~\cite[pp.~25, Sec.~2]{FGS01}.

\begin{fact}\label{ifseedsarestablethentheprocessismonotone}
Let $G(V,E)$ be a
directed graph,
$\rho\in(\,0,1\,],$
$t\in\mathbb{Z}^+$
and $S\subseteq V$
satisfy
\begin{eqnarray}
S&\subseteq& \text{\sf Active}^{(1)}\left(S,G,\rho\right),
\label{sayingtheseedsarestable}
\\
V&=&\text{\sf Active}^{(t)}\left(S,G,\rho\right).\nonumber
\end{eqnarray}
Then
$$\text{\rm min-seed}^\text{\rm async}_\text{\rm
monotone}\left(G,\rho\right)\le |\,S\,|.$$
\comment{
no active vertices can ever be deactivated
in any asynchronous
reversible cascade
with
seed set $S$ and threshold $\rho$.
Furthermore,
for $k\in\mathbb{N},$
$$
\text{\sf Active}^{(k)}\left(S,G,\rho\right)
\subseteq
\text{\sf Active}^{\text{\rm async}}\left(S,G,\rho\right).
$$
}
\comment{
Furthermore,
$$
V
=
\text{\sf Active}^{\text{\rm async}}\left(S,G,\rho\right).$$
}
\end{fact}

\comment{
Assuming that
active vertices remain active forever,
it is
folklore
that
the
timing
of activating the vertices does not
affect the set of vertices that
will
be
active.
Therefore,
Fact~\ref{ifseedsarestablethentheprocessismonotone}
has the following consequence, which is also folklore.

\begin{fact}\label{ifseedsarestablethensynchronizationinvariant}
Let $G(V,E)$ be a
directed graph,
$\rho\in(\,0,1\,]$ and $S\subseteq V$
satisfy
Eq.~(\ref{sayingtheseedsarestable}).
Then
for $k\in\mathbb{N},$
$$
\text{\sf Active}^{(k)}\left(S,G,\rho\right)
\subseteq
\text{\sf Active}^{\text{\rm async}}\left(S,G,\rho\right).$$
\end{fact}
}

The following
is Markov's inequality.

\begin{fact}(\cite[Theorem 3.2]{MR95})\label{Markovinequality}
Let $X$ be a random variable taking nonnegative values.
Then for each $t>0,$
$$\Pr\left[\,X\ge
t
\,\right]
\le
\frac{E[\,X\,]}{t}.$$
\end{fact}

We will use the following form of Chernoff's bound~\cite{Che52}.

\begin{fact}
(\cite[Theorem 4.2]{MR95})
\label{Chernofflowertail}
Let
$X_1,X_2,\ldots, X_k$ be independent
random variables
taking values in $\{0,1\}$
and $p\in[\,0,1\,]$
such that
$\Pr[\,X_i=1\,]=p$ for $1\le i\le k$.
Then for each $\delta\in (0,1),$
$$\Pr\left[\,\sum_{i=1}^k X_i \le (1-\delta)\,
kp\,\right]
\le \exp{\left(-\frac{\delta^2 kp}{2}\right)}
$$
\end{fact}

For $n\in\mathbb{Z}^+$ and $p\in[\,0,1\,],$
the Erd\H{o}s-R\'enyi random graph $G(n,p)$
is the simple undirected graph
with vertices
$1,2,\ldots,n$
where each of the possible $\binom{n}{2}$ edges
appears independently with probability $p$~\cite{Bol01}.
Below is an easy consequence of Chernoff's bound.

\begin{fact}(\cite[Lemma~9]{CLTOCS})\label{badvertices}
Let $n\in\mathbb{Z}^+,$ $p\in[\,0,1\,]$ and
$\kappa\in \{1/n,\ldots,n/n\}$. If
\begin{eqnarray*}
0.99\le \frac{\binom{\kappa n}{2}+\kappa n(n-\kappa n)}
{\kappa n^2} \le 1,
\end{eqnarray*}
then
$$\Pr\left[\,\left|\,\left\{v\in [\,n\,]\,\mid\,
d(v)
\le
\frac{pn}{2}\right\}\,\right| \ge
\kappa n\,\right]
\le \binom{n}{\kappa n} \exp{\left(
-\frac{\kappa pn^2}{9}\right)},$$
where the probability is taken over the Erd\H{o}s-R\'enyi
random graphs
$G(n,p)$.
\end{fact}

The
irreversible
cascade
is the
modification
of
our synchronous
reversible
cascade
that prohibits the deactivation of vertices.
It can
also
be defined
as
an asynchronous
process
without affecting the set of
vertices that will be
active~\cite[pp.~25, Sec.~2]{FGS01}.
Assuming irreversible cascades,
Chang and Lyuu~\cite{CLTOCS, Cha10} prove
the following bound on activating vertices
of Erd\H{o}s-R\'enyi
random graphs.

\begin{fact}(\cite{CLTOCS, Cha10})
\label{themainresultonirreversiblecascadesinERgraphs}
Assume irreversible cascades.
For a sufficiently large constant $\beta>0,$
$n\in\mathbb{Z}^+,$
$\rho\in(\,0,1\,],$
$\delta\in\{1/n,2/n,\ldots,n/n\},$
$p\ge \beta(\ln{(e/\min\{\delta,\rho\})})/(\rho n)$
and with probability $1-n^{-\Omega(1)}$ over $G(n,p),$
the minimum number of seeds needed to eventually activate
at least $\delta n$
vertices of $G(n,p)$ is $\Theta(\min\{\delta,
\rho\}\,n)$.
The hidden constant in the $\Theta(\cdot)$ notation
is independent of $p$.
\end{fact}

\section{General bounds}\label{generalboundssection}

We begin with several lemmas that will be
useful for analyzing
reversible cascades on scale-free graphs.

\comment{
To activate all vertices
of an undirected connected graph $G,$
the following lemma
shows that it suffices to
pick as seeds
each
$v\in V$ with $d(v)>1/\rho$
together with
at least a $\rho$ fraction of
the vertices in $N(v)$.
In proof,
we show that the seeds thus picked
remain
active forever.
Then we observe that
activating
a
non-seed
vertex
requires
the activation of only one of
its neighbors;
hence
the connectedness of $G$ implies that all vertices will
be active.
}

\begin{lemma}\label{thebasis1forscalefreegraphs}
Let $G(V,E)$ be
an undirected connected graph with diameter $\Delta,$
$|\,V\,|\ge 2$
and $\{v\in V\mid d(v)>1/\rho\}\neq \emptyset$.
Then
\begin{eqnarray}
\text{\rm min-seed}^{(\Delta)}\left(G,\rho\right)
\le
\sum_{v\in V, d(v)>1/\rho}\, \left(
\left\lceil\rho\,d(v)\right\rceil+1
\right).
\label{theboundforthebasisforscalefreegraphs}
\end{eqnarray}
\end{lemma}
\begin{proof}
For each $v\in V,$
pick
a
set $B(v)\subseteq N(v)$ of size
$\lceil\rho\,d(v)\rceil$.
Let
\begin{eqnarray}
X
=
\left\{v\in V\mid d(v)>\frac{1}{\rho}\right\}.
\label{thecoreofseeds}
\end{eqnarray}
Then
take
$$
S=
\bigcup_{v\in X}\,
\left(B(v)\cup\{v\}\right)$$
as the set of seeds.
Clearly, $S\neq \emptyset$.
As the righthand side of Eq.~(\ref{theboundforthebasisforscalefreegraphs})
is an upper bound on $|\,S\,|,$
it suffices to
establish $\text{\sf Active}^{(\Delta)}(S,G,\rho)=V$.

For any $u\in S,$
we have
\begin{eqnarray}
\left|\,N(u)\cap S\,\right|
\ge \lceil\rho\,d(u)\rceil
\label{thetakenseedsarestable}
\end{eqnarray}
by the following arguments:
\begin{itemize}
\item If
$u\in X,$
then $B(u)\subseteq S,$
implying Eq.~(\ref{thetakenseedsarestable}).
\item Otherwise, $u\in B(v)$ for some
$v\in X$.
Therefore,
$v\in N(u)\cap S,$ implying
Eq.~(\ref{thetakenseedsarestable})
because
$d(u)\le 1/\rho$.
\end{itemize}
By Eq.~(\ref{thetakenseedsarestable}),
if
all
vertices in $S$
are active in a round, then they will remain to be active in the
next round.
Therefore,
as the vertices in $S$ are active in round zero,
\begin{eqnarray}
S\subseteq \bigcap_{k\ge 0}\,\text{\sf Active}^{(k)}\left(S,G,\rho\right).
\label{thebasicseedsarestable}
\end{eqnarray}
For each $i\ge 0,$
denote by
$P(i)$
the relation
\begin{eqnarray}
N^i[S]\subseteq \bigcap_{k\ge i}\, \text{\sf
Active}^{(k)}\left(S,G,\rho\right).
\label{thesteppingofactivation}
\end{eqnarray}
By construction, every $w\in V\setminus S$
satisfies $d(w)\le 1/\rho$
and thus
\begin{eqnarray}
\left\lceil\rho\,d(w)\right\rceil=1.
\label{nonseedscanbeactivatedbyjustoneactiveneighbor}
\end{eqnarray}
For $i\in\mathbb{N},$
$P(i)$ implies $P(i+1)$
by the following arguments:
\begin{itemize}
\item Every $w\in N^{i+1}[S]\setminus N^i[S]$
has a neighbor in $N^i[S]$.
So by $P(i),$
$w$ has at least one active neighbor in rounds $i,i+1,\ldots$
Hence by Eq.~(\ref{nonseedscanbeactivatedbyjustoneactiveneighbor})
and the definition of the synchronous reversible cascade,
$w$ is active in rounds $i+1,i+2,\ldots$
\item By $P(i),$
all vertices in $N^i[S]$ are active in rounds $i+1,i+2,\ldots$
\end{itemize}
As Eq.~(\ref{thebasicseedsarestable})
is precisely
$P(0),$
$P(i)$ holds for all $i\in\mathbb{N}$ by mathematical induction.
Finally, $P(\Delta)$
gives $\text{\sf Active}^{(\Delta)}(S,G,\rho)=V$.
\comment{
By construction, every $w\in V\setminus S$ has degree less than
or equal to $1/\rho$ and thus $\lceil\rho\,d(w)\rceil=1$.
Therefore,
in the first round,
each vertex
with distance $1$ to a vertex in $S$
will
be
active.
}
\end{proof}

The following lemma
reduces the
bound of Lemma~\ref{thebasis1forscalefreegraphs}
to $2$
if all degrees are less than or equal to $\rho$.

\begin{lemma}\label{thebasis2forscalefreegraphs}
Let $G(V,E)$ be
an undirected connected graph with diameter $\Delta,$
$|\,V\,|\ge 2$
and $\{v\in V\mid d(v)>1/\rho\}=\emptyset$.
Then
\begin{eqnarray}
\text{\rm min-seed}^{(\Delta)}\left(G,\rho\right)
\le
2.
\label{theboundis2wheneveryoneiseasytotrigger}
\end{eqnarray}
\end{lemma}
\begin{proof}
Modify the proof of Lemma~\ref{thebasis1forscalefreegraphs}
by taking
$X$ to contain a single vertex in Eq.~(\ref{thecoreofseeds})
and replacing
the
reference to
``Eq.~(\ref{theboundforthebasisforscalefreegraphs})''
by ``Eq.~(\ref{theboundis2wheneveryoneiseasytotrigger}).''
The
rest of the proof
follows,
word for word.
\end{proof}

\begin{lemma}\label{thebasisforscalefreegraphs}
Let $G(V,E)$ be
an undirected connected graph with diameter $\Delta$ and
$|\,V\,|\ge 2$.
Then
$$\text{\rm min-seed}^{(\Delta)}\left(G,\rho\right)
\le
2
+
\sum_{v\in V, d(v)>1/\rho}\, \left(
\left\lceil\rho\,d(v)\right\rceil+1
\right).
$$
\end{lemma}
\begin{proof}
Immediate from
Lemmas~\ref{thebasis1forscalefreegraphs}--\ref{thebasis2forscalefreegraphs}.
\end{proof}

\begin{corollary}\label{thebasisforscalefreegraphsthemonotoneversion}
Let $G(V,E)$ be
an undirected connected graph with diameter $\Delta$ and
$|\,V\,|\ge 2$.
Then
$$\text{\rm min-seed}^{\text{\rm async}}_\text{\rm monotone}\left(G,\rho\right)
\le
2
+
\sum_{v\in V, d(v)>1/\rho}\, \left(
\left\lceil\rho\,d(v)\right\rceil+1
\right).$$
\end{corollary}
\begin{proof}
Eq.~(\ref{thetakenseedsarestable})
holds in our proofs of
Lemmas~\ref{thebasis1forscalefreegraphs}--\ref{thebasis2forscalefreegraphs},
and
it
implies
Eq.~(\ref{sayingtheseedsarestable}).
Hence
the corollary follows from
Fact~\ref{ifseedsarestablethentheprocessismonotone}
and
the proofs of
Lemmas~\ref{thebasis1forscalefreegraphs}--\ref{thebasis2forscalefreegraphs}.
\end{proof}

\comment{
It is folklore that if
the
initial seeds
remain active
forever,
then no vertices will ever be deactivated.
To see this in the case of Lemma~\ref{thebasisforscalefreegraphs},
}

\comment{
We remark on the proof of Lemma~\ref{thebasisforscalefreegraphs}.
As a vertex can be activated only by its neighbors,
\begin{eqnarray}
\text{\sf Active}^{(i)}\left(S,G,\rho\right)\subseteq N^i[S].
\label{nofarawayverticesshouldbeactivatedsoon}
\end{eqnarray}
By Eq.~(\ref{thesteppingofactivation})
Eq.~(\ref{nofarawayverticesshouldbeactivatedsoon})
holds containment replaced by equality.
Consequently,
the seeds form a monotone dynamic monopoly
in the sense that
no active vertices will ever be deactivated~\cite{Pel98, FKRRS03, FLLPS04,
Ber01}.
}

Next, we analyze
the following picking of seeds in a
directed graph
$G(V,E)$:
First,
choose each vertex as a seed independently with some probability.
Second,
choose
the vertices
that cannot be activated in round $1$
also as seeds.
The
expected number of chosen seeds
serves as an upper bound on $\text{min-seed}^{(1)}(G,\rho)$.
Chang and Lyuu~\cite[Theorem~9]{CL09}
use a similar technique
to
analyze irreversible cascades.
Their
results
are
improved by Ackerman, Ben-Zwi
and Wolfovitz~\cite[Sec.~3]{ABW10}.

\begin{lemma}
For
a
directed graph
$G(V,E),$
$\rho\in (\,0,1\,]$
and
$C>1,$
\begin{eqnarray}
\text{\rm min-seed}^{(1)}\left(G,\rho\right)
\le O\left(C\rho\,|\,V\,|\right)
+\sum_{v\in V}\, \left(d^\text{\rm in}(v)+1\right)
\exp{\left(-3C\rho\, d^\text{\rm in}(v)\right)}.
\label{aneasyboundthatbecomessogoodforERgraphs}
\end{eqnarray}
\end{lemma}
\begin{proof}
Assume $\rho<1/(8C)$ for, otherwise,
Eq.~(\ref{aneasyboundthatbecomessogoodforERgraphs})
holds trivially.
Let $S\subseteq V$ contain
each vertex
independently with probability
$8C\rho$
and
$$A=\left\{v\in V\mid \left|\,N^\text{in}(v)\cap S\,\right|
\le
\rho\,d^\text{in}(v)\right\}.$$
By construction,
$V\setminus A\subseteq\text{\sf Active}^{(1)}(S,G,\rho)$.
Clearly,
$A\subseteq \text{\sf Active}^{(1)}(N^\text{in}(A)\cup A, G, \rho)$.
Consequently,
\begin{eqnarray}
\text{\sf Active}^{(1)}\left(S\cup N^\text{in}(A)\cup A, G, \rho\right)=V.
\label{thiskindofseedingsuffices}
\end{eqnarray}

Clearly,
\begin{eqnarray}
E\left[\,|\,S\,|\,\right]&=&8C\rho\, |\,V\,|,\label{thefirstseeds}\\
\left|\,N^\text{in}(A)\cup A\,\right|
&\le&
\sum_{v\in A}\, \left(d^\text{in}(v)+1\right).
\label{takeneighborhoodofeachinactivevertex}
\end{eqnarray}
By Chernoff's bound (Fact~\ref{Chernofflowertail}),
\begin{eqnarray}
\Pr\left[\,
v\in A
\,\right]
\le \exp{\left(-3C\rho\,d^\text{in}(v)\right)}
\label{probabilityofnotactivatedin1round}
\end{eqnarray}
for all $v\in V$.
By
Eqs.~(\ref{takeneighborhoodofeachinactivevertex})--(\ref{probabilityofnotactivatedin1round})
and the linearity of expectation,
$$E\left[\,\left|\,N^\text{in}(A)\cup A\,\right|\,\right]
\le \sum_{v\in V}\,\left(d^\text{in}(v)+1\right)\,
\exp{\left(-3C\rho\,d^\text{in}(v)\right)}.$$
This and Eqs.~(\ref{thiskindofseedingsuffices})--(\ref{thefirstseeds})
complete the proof because there must exist a realization
of $S\cup N^\text{in}(A)\cup A$ with size less than or equal to
its expected value.
\end{proof}

\begin{lemma}\label{justdontknowhowtonamethislemma}
For
a
directed graph
$G(V,E),$
$\rho\in (\,0,1\,]$
and
$C>1,$
$$\text{\rm min-seed}^{(1)}\left(G,\rho\right)
\le O\left(C\rho\,|\,V\,|\right)
+
\sum_{v\in V, d^\text{\rm in}(v)< (1/(C\rho))\ln(e/\rho)}\,
\left(d^\text{\rm in}(v)+1\right)
\exp{\left(-3C\rho\, d^\text{\rm in}(v)\right)}.$$
\end{lemma}
\begin{proof}
By elementary calculus,
$$\max_{x\ge (1/(C\rho))\ln(e/\rho)}\,(x+1)\exp(-3C\rho x)=O(\rho).$$
Therefore,
in the summation of
Eq.~(\ref{aneasyboundthatbecomessogoodforERgraphs}),
vertices with indegrees at least $(1/(C\rho))\ln(e/\rho)$
contribute $O(\rho\, |\,V\,|)$ in total.
\end{proof}

\begin{lemma}\label{aboundinthenumberofsmalldegreevertices}
For
a
directed graph
$G(V,E),$
$\rho\in (\,0,1\,]$
and
$C>1,$
$$\text{\rm min-seed}^{(1)}\left(G,\rho\right)
\le O\left(C\rho\,|\,V\,|\right)
+
\left|\,\left\{v\in V\mid d^\text{\rm in}(v)
<\frac{1}{C\rho}\ln{\frac{e}{\rho}}\right\}\,\right|
\cdot\left(\frac{1}{C\rho}\left(\ln{\frac{e}{\rho}}\right)+1\right).$$
\end{lemma}
\begin{proof}
Invoke Lemma~\ref{justdontknowhowtonamethislemma}
and observe that
$\exp{(-3C\rho\, d^\text{in}(v))}\le 1$ for all $v\in V$.
\end{proof}



\begin{corollary}\label{simpleconsequenceforregulargraphs}
Let
$G(V,E)$ be a directed graph
$\rho\in (\,0,1\,]$
and
$C>1$.
If every vertex of $G$ has indegree
$\Omega((1/\rho)\,\ln{(e/\rho)}),$
then
$$\text{\rm min-seed}^{(1)}\left(G,\rho\right)
= O\left(\rho\,|\,V\,|\right).$$
\end{corollary}
\begin{proof}
Invoke Lemma~\ref{aboundinthenumberofsmalldegreevertices}
with a sufficiently large constant $C$.
\end{proof}

Corollary~\ref{simpleconsequenceforregulargraphs}
holds, e.g., for
$\Omega((1/\rho)\,\ln{(e/\rho)})$-regular graphs.

\section{Bounds for connected scale-free graphs}
\label{connectedscalefreegraphssection}

Let $G(V,E)$ be a connected scale-free graph
with diameter $\Delta$
and $2<\gamma<3$ be a constant
such that
$G$ has an $O(1/k^\gamma)$ fraction of
degree-$k$ vertices,
$k\in\mathbb{Z}^+$.
This section shows that
$\text{min-seed}^{(\Delta)}(G,\rho)
=O(\lceil\rho^{\gamma-1}\,|\,V\,|\rceil),$
which suggests rapid
reversible
cascades
because scale-free graphs typically have
small diameters or
average vertex-vertex
distances~\cite{NSW01, CH03, BR04, CL04a, CL04b}.
In the case of asynchronous reversible cascades,
$\text{min-seed}^\text{async}_\text{monotone}(G,\rho)
=O(\lceil\rho^{\gamma-1}\,|\,V\,|\rceil)$.
Compared with
Fact~\ref{themainresultonirreversiblecascadesinERgraphs},
therefore,
activating all vertices
requires fewer seeds
on
connected scale-free graphs
than
on
Erd\H{o}s-R\'enyi random graphs
$G(n,p)$
for
$\omega(1/n)\le\rho\le o(1)$
and $p=\Omega((\ln(e/\rho))/(\rho n))$
with a sufficiently large hidden constant in the
$\Omega(\cdot)$ notation.
This holds
even if
we
allow deactivations
and
require the cascades to succeed
regardless of the (progressive) synchronizations
in connected scale-free graphs
but not in Erd\H{o}s-R\'enyi random graphs.

\comment{
Given
Fact~\ref{themainresultonirreversiblecascadesinERgraphs},
it is not clear whether
$\text{min-seed}^\text{async}(G,\rho)=O(\rho\,|\,V\,|)$.
On the one hand,
allowing deactivations of vertices
may make it less likely for a set of seeds to
activate all vertices
than in Fact~\ref{themainresultonirreversiblecascadesinERgraphs}.
Furthermore,
proving
$\text{min-seed}^\text{async}(G,\rho)=O(\rho\,|\,V\,|)$
requires
the existence of $O(\rho\,|\,V\,|)$ seeds that
activate all vertices no matter how activations and
deactivations are synchronized.
On the other hand,
the
distribution of
degrees
in scale-free graphs
may enable
bounding
$\text{min-seed}^\text{async}(G,\rho)$
in a way impossible for Erd\H{o}s-R\'enyi random graphs.
This section shows that
$\text{min-seed}^\text{async}(G,\rho)
=O(\lceil\rho^{\gamma-1}\,|\,V\,|\rceil)$.
In the case of synchronous reversible cascades
and writing $\Delta$ for the diameter of $G,$
$\text{min-seed}^{(\Delta)}(G,\rho)
=O(\lceil\rho^{\gamma-1}\,|\,V\,|\rceil),$
which suggests rapid global cascades
because
scale-free graphs typically have
small
diameters
or average vertex-vertex
distances~\cite{NSW01, CH03, BR04, CL04a, CL04b}.
}

\begin{theorem}\label{maintheoremforscalefreegraphs}
Let
$G(V,E)$ be an undirected connected graph
with diameter $\Delta$
and
$\rho\in(\,0,1\,]$.
If
$C>0$ and $\gamma>2$
are constants
with
\begin{eqnarray}
\frac{\left|\,\left\{v\in V\mid d(v)=k\right\}\,\right|}{|\,V\,|}\le
\frac{C}{k^\gamma}
\label{thescalefreecondition}
\end{eqnarray}
for all $k\in\mathbb{Z}^+,$
then
\begin{eqnarray}
\text{\rm min-seed}^{(\Delta)}\left(G,\rho\right)
=O\left(\left\lceil\rho^{\gamma-1}\,|\,V\,|\right\rceil\right).
\label{theboundinthescalefreecase}
\end{eqnarray}
\end{theorem}
\begin{proof}
We may assume without loss of generality that $\rho<0.1$
for, otherwise, Eq.~(\ref{theboundinthescalefreecase}) holds trivially (note
that $\gamma$ is a constant).
By Lemma~\ref{thebasisforscalefreegraphs},
\begin{eqnarray}
\text{\rm min-seed}^{(\Delta)}\left(G,\rho\right)
=O\left(
2
+
\sum_{v\in V, d(v)>1/\rho}\, \rho\,d(v)
\right).
\label{justwritingtheboundinaslightlymorereadableform}
\end{eqnarray}
Now,
\begin{eqnarray*}
&&\sum_{v\in V, d(v)>1/\rho}\, d(v)\\
&\stackrel{\text{Eq.~(\ref{thescalefreecondition})}}{\le}&
\sum_{k\in\mathbb{N}, k>1/\rho}\,
k\cdot\frac{C\,|\,V\,|}{k^\gamma}\\
&\le& \int_{1/\rho-1}^\infty\! \frac{C\,|\,V\,|}{x^{\gamma-1}}\,
\mathrm{d}x\\
&=&
O\left(\rho^{\gamma-2}|\,V\,|\right),
\end{eqnarray*}
where the second inequality follows from elementary calculus
and the $O(\cdot)$ notation hides constants dependent on
$C$ and $\gamma$.
This and Eq.~(\ref{justwritingtheboundinaslightlymorereadableform})
complete the proof.
\end{proof}

\begin{corollary}
Let
$G(V,E)$ be an undirected connected graph
with diameter $\Delta$ and $\rho\in(\,0,1\,]$.
If $C>0$ and $\gamma>2$ are constants
satisfying Eq.~(\ref{thescalefreecondition})
for all $k\in\mathbb{Z}^+,$
then
\begin{eqnarray}
\text{\rm min-seed}^{\text{\rm async}}_\text{\rm monotone}\left(G,\rho\right)
=O\left(\left\lceil\rho^{\gamma-1}\,|\,V\,|\right\rceil\right).
\end{eqnarray}
\end{corollary}
\begin{proof}
In the proof of
Theorem~\ref{maintheoremforscalefreegraphs},
invoke Corollary~\ref{thebasisforscalefreegraphsthemonotoneversion}
instead of Lemma~\ref{thebasisforscalefreegraphs}.
\end{proof}

\section{Bounds for Erd\H{o}s-R\'enyi random graphs}
\label{ErdosRenyirandomgraphssection}

Fact~\ref{themainresultonirreversiblecascadesinERgraphs}
does not give
that
\begin{eqnarray*}
\text{min-seed}^{(1)}\left(G(n,p),\rho\right)=O\left(\rho n\right),\\
\text{min-seed}^{\text{async}}_\text{monotone}\left(G(n,p),\rho\right)=O\left(\rho n\right)
\end{eqnarray*}
because of possible deactivations,
the requirement of activating all vertices
in only one round
and the arbitrary (yet progressive) synchronizations of
the cascades.
Still, this section proves both equations
for $p=\Omega((\ln(e/\rho))/(\rho n))$.

\begin{lemma}\label{takingcareofalargevalueofrho}
Let
$n\in\mathbb{Z}^+,$
$\rho\in [\,1/n^{1/3},1\,],$
$C>1$
and
$p\ge (100/(C\rho n))\ln(e/\rho)$.
Then with probability
$1-\exp{(-\Omega((\rho^2 n/C)\ln{(e/\rho)}))}$
over the Erd\H{o}s-R\'enyi random graphs $G(n,p),$
\begin{eqnarray}
\text{\rm min-seed}^{(1)}\left(G(n,p),\rho\right)
=O\left(
C\rho n
\right).
\label{thisistheboundforERgraphs}
\end{eqnarray}
\end{lemma}
\begin{proof}
Assume $\rho<1/(10C)$ for, otherwise, Eq.~(\ref{thisistheboundforERgraphs})
holds trivially.
By Fact~\ref{badvertices},
\begin{eqnarray}
\Pr\left[\,\left|\,\left\{v\in V\mid d(v)
<\frac{1}{C\rho}\ln{\frac{e}{\rho}}\right\}\,\right|
\ge
\left\lfloor
\rho^3 n
\right\rfloor
\,\right]
\le \exp{\left(
-\Omega\left(\frac{\rho^2 n}{C}\ln{\frac{e}{\rho}}
\right)\right)}.
\nonumber
\end{eqnarray}
This
and
Lemma~\ref{aboundinthenumberofsmalldegreevertices}
complete
the proof.
\end{proof}

\begin{lemma}\label{takingcareofsmallvaluesofrho}
Let
$n\in\mathbb{Z}^+,$
$\rho\in (\,0,1\,],$
$C>1$
and
$p\ge (100/(C\rho n))\ln(e/\rho)$.
Then with probability
$1-(1/\rho^3)\exp{(-\Omega((1/(C\rho))\ln{(e/\rho)}))}$
over the Erd\H{o}s-R\'enyi random graphs $G(n,p),$
$$\text{\rm min-seed}^{(1)}\left(G(n,p),\rho\right)
=O\left(C\rho n\right).$$
\end{lemma}
\begin{proof}
By Chernoff's bound (Fact~\ref{Chernofflowertail}),
$$\Pr\left[\,d(v)<\frac{1}{C\rho}\ln{\frac{e}{\rho}}\,\right]
\le
\exp{\left(-\Omega\left(\frac{1}{C\rho}\ln{\frac{e}{\rho}}\right)\right)},
$$
$v\in\{1,2,\ldots,n\}$.
So by the linearity of expectation,
$$E\left[\,\left|\,\left\{v\in
\{1,2,\ldots,n\}
\mid d(v)
<\frac{1}{C\rho}\ln{\frac{e}{\rho}}\right\}\,\right|\,\right]
\le
n\,\exp{\left(-\Omega\left(\frac{1}{C\rho}\ln{\frac{e}{\rho}}\right)\right)},
$$
implying
$$
\Pr\left[\,\left|\,\left\{v\in
\{1,2,\ldots,n\}
\mid d(v)
<\frac{1}{C\rho}\ln{\frac{e}{\rho}}\right\}\,\right|\ge \rho^3 n\,\right]
\le
\frac{1}{\rho^3}\,
\exp{\left(-\Omega\left(\frac{1}{C\rho}\ln{\frac{e}{\rho}}\right)\right)}
$$
by Markov's inequality (Fact~\ref{Markovinequality}).
This and Lemma~\ref{aboundinthenumberofsmalldegreevertices}
complete the proof.
\end{proof}

\begin{theorem}\label{themainresultforERgraphs}
Let $n\in\mathbb{Z}^+,$
$\rho\in (\,0,1\,]$ and
$p=\Omega((\ln(e/\rho))/(\rho n))$.
Then with probability $1-\exp{(-n^{\Omega(1)})}$
over the Erd\H{o}s-R\'enyi random graphs $G(n,p),$
$$\text{\rm min-seed}^{(1)}\left(G(n,p),\rho\right)
=O\left(\rho n\right).$$
\end{theorem}
\begin{proof}
For a sufficiently large constant $C,$ invoke
Lemma~\ref{takingcareofalargevalueofrho}
if $\rho\ge 1/n^{1/3}$
and Lemma~\ref{takingcareofsmallvaluesofrho}
otherwise.
\end{proof}

By
Fact~\ref{themainresultonirreversiblecascadesinERgraphs}
and
Theorem~\ref{themainresultforERgraphs},
for $p=\Omega((\ln(e/\rho))/(\rho n))$
with a sufficiently large hidden constant in the
$\Omega(\cdot)$ notation,
$$\text{\rm min-seed}^{(k)}\left(G(n,p),\rho\right)
=\Theta\left(\rho n\right)$$
with probability $1-n^{-\Omega(1)}$
for all
$k\in\mathbb{Z}^+$.
This is asymptotically the same as the
bound in Fact~\ref{themainresultonirreversiblecascadesinERgraphs}
for irreversible cascades.

\comment{
\begin{corollary}
There exists a constant $\beta>0$ such that for $n\in\mathbb{Z}^+,$
$\rho\in (\,0,1\,],$
$p\ge\beta(\ln(e/\rho))/(\rho n)$
and with probability $1-n^{-\Omega(1)}$
over $G(n,p),$
$$\text{\rm min-seed}^{(k)}\left(G(n,p),\rho\right)
=\Theta(\rho n)$$
for all $k\in\mathbb{Z}^+$.
\end{corollary}
}

\comment{
For $p=\Omega((\ln{(e/\min\{\delta,\rho\})})/(\rho n))$
with a sufficiently large hidden constant in the
$\Omega(\cdot)$ notation,
the $O\left(\rho n\right)$ bound in
Theorem~\ref{themainresultforERgraphs}
is asymptotically optimal
by Fact~\ref{themainresultonirreversiblecascadesinERgraphs}.
}

\begin{corollary}
Let $n\in\mathbb{Z}+,$ $\rho\in (\,0,1\,]$ and
$p=\Omega((\ln(e/\rho))/(\rho n))$.
Then with probability $1-\exp{(-n^{\Omega(1)})}$
over the Erd\H{o}s-R\'enyi random graphs $G(n,p),$
$$\text{\rm min-seed}^{\text{\rm async}}_\text{\rm monotone}\left(G(n,p),\rho\right)
=O\left(\rho n\right).$$
\end{corollary}
\begin{proof}
By Theorem~\ref{themainresultforERgraphs},
there exists a set $S$ of $O(\rho n)$ seeds
with $\text{\sf Active}^{(1)}(S,G,\rho)=V$.
Hence the theorem follows from
Fact~\ref{ifseedsarestablethentheprocessismonotone}.
\end{proof}

\section{Conclusions}\label{conclusionssection}

Reversible cascades are central in
local interaction games with full rationality~\cite{Mor00, Kle07}
and the propagation of transient faults in majority-based
systems~\cite{Pel02, FKRRS03, FLLPS04}.
We investigated
them
in connected scale-free and
Erd\H{o}s-R\'enyi random graphs.
Suppose
$\omega(1/n)\le\rho\le o(1)$
and $p=\Omega((\ln(e/\rho))/(\rho n))$
with a sufficiently large hidden constant in the
$\Omega(\cdot)$ notation.
Theorems~\ref{maintheoremforscalefreegraphs},~\ref{themainresultforERgraphs}
and
Fact~\ref{themainresultonirreversiblecascadesinERgraphs}
show that
activating all vertices
is easier (in terms of the number of seeds deployed)
on
connected scale-free graphs
than
on
Erd\H{o}s-R\'enyi random graphs $G(n,p)$.
However, an asymptotically smallest set of seeds
that activate
all
vertices of $G(n,p)$ (which has size $\Theta(\rho n)$
by Fact~\ref{themainresultonirreversiblecascadesinERgraphs}
and
Theorem~\ref{themainresultforERgraphs})
can do so in one round
by
Theorem~\ref{themainresultforERgraphs},
whereas
Theorem~\ref{maintheoremforscalefreegraphs}
does not show
seeds
activating
all vertices of
connected scale-free
graphs
within $O(1)$
rounds.
It
thus
remains
to
further
investigate
the tradeoff
between the number of seeds and the round complexity
for activating all vertices in scale-free graphs.
Such tradeoffs are studied
for many graphs by Flocchini et al.~\cite{FKRRS03}
because of their theoretical and practical importance.
Another interesting direction
is to refine the $O(\lceil\rho^{\gamma-1}|\,V\,|\rceil)$ bound
in Theorem~\ref{maintheoremforscalefreegraphs}
for
many
important
scale-free graphs such as those mentioned in
Sec.~\ref{introductionsection}.



\bibliographystyle{plain}
\bibliography{revER}


\end{document}